\documentclass[runningheads]{llncs}

\usepackage{fontenc}[T1]
\usepackage{alltt}
\usepackage{graphicx}

\usepackage[english]{babel}
\usepackage{multirow}
\graphicspath{{figures/}}
\usepackage{hyperref}
\usepackage{amsthm}

\usepackage{multirow}
\usepackage{booktabs}
\usepackage{amsmath}
\usepackage{amsfonts}
\usepackage{amssymb}
\usepackage{hyperref}
\usepackage{mdwlist}
\usepackage[ruled,vlined,linesnumbered]{algorithm2e}
\usepackage{color}
\usepackage{pgf}
\usepackage{longtable}
\usepackage{pdflscape}
\usepackage[round]{natbib} 
\usepackage{mathrsfs}
\usepackage{subcaption}
\usepackage{graphicx}
\usepackage{algorithmic}
\usepackage{rotating}

\begin{document}

\title{The Stackelberg Kidney Exchange Problem is $\Sigma_2^p$-complete}
\author{B. Smeulders\inst{1}, D.A.M.P. Blom\inst{1}, F.C.R. Spieksma\inst{1} }
\date{\today}\institute{Department of Mathematics and Computer Science, Eindhoven University of Technology}

\maketitle
\begin{abstract}
We introduce the Stackelberg kidney exchange problem. In this problem, an agent (e.g. a hospital or a national organization) has control over a number of incompatible patient-donor pairs whose patients are in need of a transplant. The agent has the opportunity to join a collaborative effort which aims to increase the maximum total number of transplants that can be realized. However, the individual agent is only interested in maximizing the number of transplants within the set of patients under its control. Then, the question becomes which patients to submit to the collaborative effort. We show that, whenever we allow exchanges involving at most a fixed number $K \ge 3$ pairs, answering
this question is $\Sigma_2^p$-complete. However, when we restrict ourselves to pairwise exchanges only, the problem becomes solvable in polynomial time.
\end{abstract}
\section{Introduction}

Kidney Exchange Programmes (KEPs) play a growing role in the improvement of lives of many patients suffering from end stage renal disease. Given the absence of artificial kidneys, given the lack of kidneys coming from deceased donors, and the fact that it is possible to lead a normal life with a single kidney, donation from living donors is an increasingly popular option. For this option to succeed, compatibility of donor and patient (in terms of blood type and immunological properties) is crucial.
Patients who have a willing, but incompatible, donor can be helped through a Kidney Exchange Program. In such a program, patient-donor couples, referred to as {\em pairs}, are present; the donor of a pair is willing to donate her/his kidney to some patient, provided that the corresponding patient receives a kidney from some donor. Such programmes have been established around the world; we refer to \cite{biroetal2019} for a recent overview of this practice in Europe.

The organization of collaboration between different transplant centres, organizations or countries (we use the term {\em agent} for an entity in control of a set of pairs) is a delicate matter. Such collaborations face many challenges, varying from legal considerations to the alignment of medical procedures. On one hand, it is clear that collaboration increases the possibilities for matching donors with patients, and thus leads to more transplants and better overall patient outcomes. On the other hand, these benefits of cooperation may be shared unequally, and in some cases individual agents may even lose transplants when combining patient pools.

\begin{figure}[ht!]
  \centering
  \includegraphics[width=75pt]{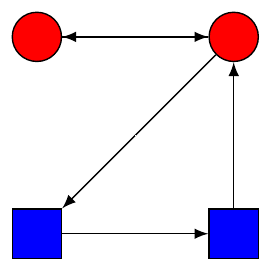}
  \caption{There exists no mechanism guaranteeing an individually rational social optimum}
  \label{Fig:Simple co-op}
\end{figure}

Consider as an illustration the example from \cite{ashlagi2014free}, depicted in Figure~\ref{Fig:Simple co-op}. Each node in this graph corresponds to a patient-donor couple, and an arc from one node to another means that the donor of the first node is compatible with the patient from the second node. A set of node-disjoint cycles in this graph corresponds to a set of realizable transplants. In Figure~\ref{Fig:Simple co-op}, two agents, red and blue, each have a private pool consisting of two pairs. Red can transplant two patients internally, while blue has no internal matches. However, if both agents combine their pools, three transplants are possible, two for blue patients and one for red. It is thus in red's best interest not to participate in the collaboration.

The example above illustrates an important result: no mechanism exists guaranteeing a solution that is both socially optimal (meaning delivering a maximum number of transplants) and individually rational (meaning that each agent acts solely in its own interest). Indeed, to give red an incentive to contribute any of its pairs, red must be guaranteed at least two transplants in the overall solution. However, any solution in which red receives two transplants is not socially optimal.

In case multiple agents can collaborate in a KEP, each agent is thus faced with the question whether to participate and if so, which of its pairs to submit to the common pool.

Issues surrounding collaboration in KEPs have been studied in the literature. \cite{ashlagi2014free} and \cite{toulis2015design} study cooperation in large random graphs, where only blood type compatibility is a limiting factor; they find that individually rational solutions exist that are close to socially optimal. 
\cite{blum2017opting} obtain similar results for arbitrary graphs. They show that with high probability, individually rational solutions are socially optimal. Additionally, they show that any socially optimal solution is with high probability close to individually rational.

Strategy-proof mechanisms for multi-agent kidney exchange are studied in \cite{ashlagi2015}. They show that even for two players and a maximum cycle length of 2, no strategy proof socially optimal mechanism can exist. They furthermore prove approximation bounds for deterministic and randomized mechanisms and propose a strategy-proof randomized mechanism that guarantees half of the optimal social welfare in the worst case. \cite{caragiannis2015improved} strengthen the bounds for randomized mechanisms and propose a strategy-proof randomized mechanism guaranteeing two-thirds of the maximum social welfare. For multi-period settings, \cite{hajaj2015} describe a strategy proof and socially optimal mechanism based on a credit system. However, this mechanism requires knowledge of the expected arrival rate of pairs for each agent. \cite{agarwal2018market} describe a credit system where agents are rewarded for adding pairs based on the expected marginal added transplants of adding that pair to a common pool.

\cite{carvalho2017} and \cite{carvalho2019game} study collaboration in KEPs as a non-cooperative game. They show that, when the cycle length equals 2, there exist a socially optimal Nash-equilibrium, and that this equilibrium can be computed in polynomial time. This problem is closely related to our setting, and we will elaborate on the similarities and differences with our setting in Section~\ref{sec:problem}.
In this paper, we study the problem faced by an individual agent in a collaborative KEP. How easy, or how hard, is it for an agent to determine its individual rational strategy? Specifically, we consider the situation where an agent must decide which of its pairs to match internally, and which pairs to add to a common pool. In this common pool, the number of transplants is then maximized. The agent's goal is to maximize the number of its own pairs that are transplanted. Clearly, this is a relevant problem. If an agent is not able to efficiently compute strategies that maximize its number of transplants, the design of mechanisms that guide collaboration in KEPs can be affected. For instance, such a mechanism can safely assume that agents are not able to efficiently identify such strategies.

In the next section, we formally define the problem and describe our main result. We will show that, even if an agent knows exactly which pairs of other agents are present in the common pool, together with their respective compatibilities, the problem of deciding which pairs to contribute and which to match internally, to guarantee a given number of its pairs are transplanted, is $\Sigma_2^p$-complete. The class $\Sigma_2^p$ is a complexity class of decision problems that generalizes the traditional classes P and NP to a setting with two decision makers (or players/agents). It contains problems that can be expressed by a logical formula using two consecutive quantifiers, where the first quantifier is of the type ``does there exist'', while the second quantifier is of the type ``for all''.
We refer to \cite{arorabarak} for an introduction into computational complexity including the polynomial hierarchy. One practical implication of a problem being $\Sigma_2^p$-complete is that the existence of a compact Integer Program modelling the problem is unlikely, see \cite{lodiralphswoe}. Hence, our result implies that is very hard for an agent who is solely interested in maximizing the number of transplants among its own patients, to determine which patient-donor pairs to submit to the common pool, and which not. Nevertheless, whenever we restrict ourselves to two-way exchanges only, we prove that the problem becomes polynomially solvable.

\section{The problem} \label{sec:problem}
We consider simple directed graphs $G= (V,A)$. A {\em cycle} in $G$ is a set of nodes $\{v_1, v_2, \ldots, v_q\}$ such that $v_i \in V$ for $i=1, \ldots, q$, $(v_i,v_{i+1}) \in A$ for $i=1, \ldots,q-1$, and $(v_q,v_1) \in A$; the {\em length} of the cycle is its number of nodes $q$.  A {\em $K$-cycle packing} in $G$ is a set of node-disjoint cycles each of which has length at most $K$; the {\em size} of a $K$-cycle packing is the total number of nodes contained in its cycles. We use the phrase {\em cycle packing} when the length of the cycles in the cycle packing is not specified. We say a node $v \in V$ is {\em covered} by a cycle packing, if that cycle packing contains a cycle which includes node $v$. When given subsets of nodes $U,W \subseteq V$, we use $w^U(G[W])$ to denote the minimum number of nodes in $U$ that are covered in a maximum size cycle packing in $G[W]$. In case $U = W$, $w^U(g[W])$ is just the maximum size of a cycle packing in $G[W]$. For ease of notation, we denote this by $w(G[W])$.

In our formulation of the problem, we distinguish between the {\em leader} on the one hand, and the {\em follower} on the other hand. Here, the leader stands for the individual agent (i.e., country or hospital), whereas the follower stands for the organisation responsible for running the larger Kidney Exchange Program.

Both the leader and the follower are each in control of a set of nodes (recall that each node refers to a patient-donor couple), denoted by $L$ and $F$ respectively. The leader has two options for each node. One option is to withold the node from the follower, and perhaps use the node, if possible, in a cycle packing that is private to the leader. The other option is to contribute the node, thereby adding the node to the global pool of nodes under control of the follower. In the latter option it is no longer in the control of the leader which cycle packing will be selected for nodes in the global pool, however, the follower will select a maximum size cycle packing in the global pool of nodes. The question is whether the leader can guarantee that a given number of its nodes will be covered by a cycle packing? We will refer to the node set $S$ chosen by the leader to withhold from the follower, as the leader's {\em strategy} $S$.

\begin{definition}[Stackelberg KEP game]
    We define a {\bf Stackelberg KEP game} as follows.

  {\bf Given:} A directed graph $G =(V=L \cup F,A)$ with $L \cap F = \emptyset$, and integers $k$ and $K$. 
  In the first phase, the leader selects a strategy $S \subseteq L$ of nodes, and calculates a maximum size $K$-cycle packing on $G[S]$. 
  In the second phase, the follower calculates a maximum size $K$-cycle packing on $G[V \setminus S]$. 

  {\bf Question:} Is there a subset $S \subseteq L$, such that $w(G[S]) + w^{L}(G[V \setminus S]) \geq k$?
\end{definition}

Given the definition of $w^L(G)$, it follows that by considering the quantity $w(G[S]) + w^{L}(G[V \setminus S])$, we are considering the worst-case scenario for the leader. Indeed, this assumes that among all possible maximum size packings the follower can choose, the follower chooses the one covering the minimum number of nodes of the leader. This objective is reasonable for a risk-averse agent when the follower's tie-breaking rules are unknown. \\

Stackelberg KEP is closely related to, yet different from, a problem considered by \cite{carvalho2017} and \cite{carvalho2019game}. They consider a game with maximum cycle length $K = 2$ and $N \ge 2$ players in which each player $i$ simultaneously chooses a (restricted) set of nodes $S^i$, for $1 \leq i \leq N$. Each player $i$ then computes a maximum size cycle packing on $G[S^i]$ ($1 \leq i \leq N$), and an independent agent (comparable to the follower in our setting) computes a maximum size cycle packing on $G[V \setminus \bigcup_{i = 1}^N S^i]$. The goal for every player $i$ is to maximize $w(G[S^i]) + w^{V^i}(G[V \setminus \bigcup_{i = 1}^N S^i])$, ($1 \leq i \leq N$); this problem is referred to as $N$-KEG. \cite{carvalho2017} give a polynomial time algorithm for finding a socially optimal Nash equilibrium in this game if the cycle length is limited to at most two.

Apart from the fact that \cite{carvalho2017} consider an arbitrary number of players, there are two main differences between Stackelberg KEP and $N$-KEG. First, in Stackelberg KEP there are no restrictions on the strategy $S$, whereas in $N$-KEG each node in $S_i$ must be covered by a cycle packing private to player $i$. Although one might find it reasonable to consider strategies where each node that is withheld from the follower is covered, there exist instances where imposing this condition leads to strategies that are dominated. We demonstrate this phenomenon in Figure~\ref{Fig:Holdback}, where it is beneficial for the leader to not contribute nodes that are unmatched internally. Second, in Stackelberg KEP, the follower is allowed to use any cycle in its solution, while in $N$-KEG cycles containing only nodes of a single player are prohibited.


\begin{figure}[ht!]
  \centering
  \includegraphics[width=300pt]{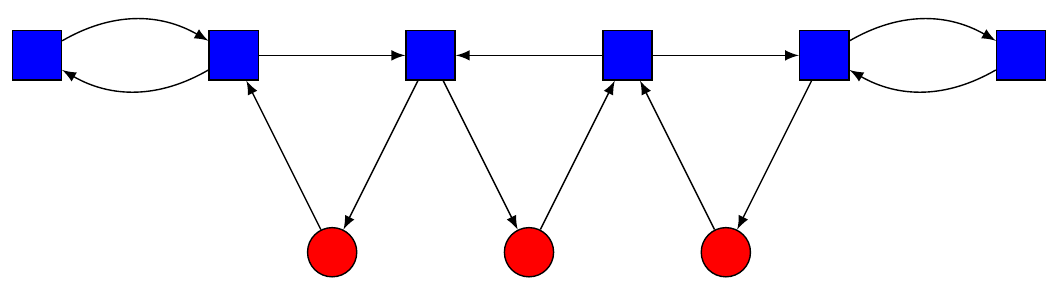}
  \caption{The leader (red circle) cannot choose any internal cycles and has incentive not to contribute its node in the middle to the global pool (wins one transplant within red patient pool).}
  \label{Fig:Holdback}
\end{figure}

Due to the close relationship between the two problems, our results for the Stackelberg game will shed light on Carvalho et al's. limitation to a maximum cycle length of 2.

\section{The result}
\label{sec:result}
We claim that Stackelberg KEP is $\Sigma_2^p$-complete. In this section, we prove this through a reduction from Adversarial (2,2)-SAT, which is defined as follows.
\begin{definition}[Adversarial (2,2)-SAT]

\noindent
{\bf Given:} Sets $X$ and $Y$ of variables, and a Boolean expression $E$ in conjunctive normal form, consisting of a set of clauses $C$, over $X$ and $Y$. Each variable occurs exactly four times in $E$: two times in negated form and two times in unnegated form. \\
\noindent
  {\bf Question:} Does there exist a truth assignment for $X$ such that there does not exist a truth assignment for $Y$ satisfying $E$?
\end{definition}

The $\Sigma_2^p$-completeness of Adversarial (2,2)-SAT follows from Theorem 2.1.1 in \cite{johannes} and through a specific reduction from 3-SAT to (2,2)-SAT given in the appendix.

We are now in a position to prove our main result.

\begin{theorem}
Stackelberg KEP is a $\Sigma_2^p$-complete problem, for each fixed $K \geq 3$.
\end{theorem}

\begin{proof}
Given an instance of Adversarial (2,2)-SAT, we construct an instance $G=(L \cup F,A)$ of Stackelberg KEP. To clearly distinguish between nodes belonging to the leader (i.e., the set $L$) and nodes belonging to the follower (i.e., the set $F$), we use Latin letters to denote the former, and Greek letters to denote the latter.

For each variable $x \in X$, we construct a gadget as depicted in Figure \ref{Fig:VarGadget}. The gadget contains four leader nodes, $\{t_{x,1},t_{x,2},f_{x,1},f_{x,2}\} \in L$ and six follower nodes, $\{\alpha_{x,1}, \alpha_{x,2}, \beta_{x,t,1}, \beta_{x,t,2},\beta_{x,f,1}, \beta_{x,t,2}\} \in F$. We denote the arcs between nodes in the gadget corresponding to $x \in X$, by the set $A_x$:
\begin{eqnarray*}
A_x &= \{(\alpha_{x,i}, \beta_{x,t,i}), (\beta_{x,t,i},t_{x,i}), (t_{x,i}, \alpha_{x,i})|~i=1,2\} \\ &\cup \{(\alpha_{x,i}, \beta_{x,f,i}), (\beta_{x,f,i},f_{x,i}), (f_{x,i}, \alpha_{x,i})|~i=1,2\} \\ &\cup \{(t_{x,1}, t_{x,2}), (t_{x,2}, t_{x,1}), (f_{x,1}, f_{x,2}), (f_{x,2}, f_{x,1})\}.
\end{eqnarray*}
For each variable $y \in Y$, we have an identical gadget, except that all its ten nodes are in $F$; to follow our naming convention, instead of nodes $t_{x,i}$ and $f_{x,i}$, the corresponding nodes in this gadget are called $\tau_{y,i}$ and $\phi_{y,i}$, $i=1,2$. The set of arcs between nodes of a gadget corresponding to $y \in Y$ is denoted by $A_y$.

For reasons of convenience, we define the set of so-called $\beta$ nodes as:
\begin{equation}
\nonumber
    B \equiv \{ \{\beta_{x,t,i}, \beta_{x,f,i}\} |~x \in X, i=1,2\} \cup \{\{\beta_{y,t,i}, \beta_{y,f,i}\}|~ y \in Y, i=1,2\}.
\end{equation}

\begin{figure}[ht!]
  \centering
  \includegraphics[width=250pt]{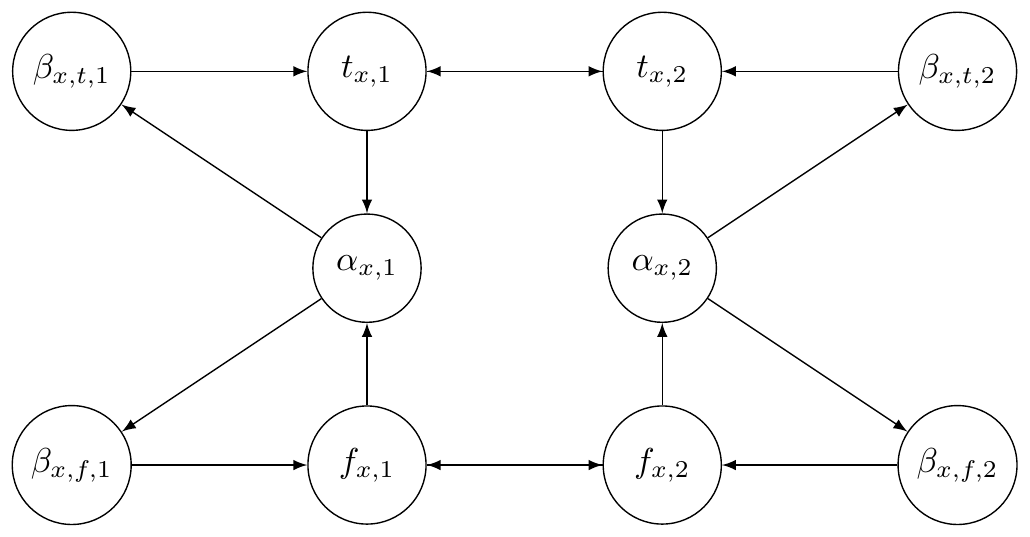}
  \caption{The gadget corresponding to a variable $x \in X$}\label{Fig:VarGadget}
\end{figure}

The construction per clause is relatively simple. For each clause $c \in C$, there exists one node, called the {\em clause node}, $\delta_c \in F$. Additionally, there is one node $d \in L$ in total. We have now specified the node sets $L$ and $F$; notice that $|L|=4|X|+1$, and $|F| = 6|X|+10|Y|+|C|.$

For each clause $c \in C$, there exist arcs $(\delta_c, d), (d,\delta_c)$ in $A$. The clause nodes are connected to the nodes in the variable gadgets as follows. For each variable $x \in X$ ($y \in Y$), there are arcs $(\beta_{x,t,1}, \delta_{c})$ ($(\beta_{y,t,1}, \delta_{c})$) and $(\delta_{c},\beta_{x,t,1})$ ($(\delta_{c},\beta_{y,t,1})$) whenever $c$ is the first clause in which $x$ ($y$) occurs unnegated, with respect to a lexicographical ordering of the clause set $C$. Analogously, the node $\beta_{x,t,2}$ ($\beta_{y,t,2}$) is connected to the clause node that corresponds to the second clause in which $x$ ($y$) occurs unnegated. Similarly, there are arcs $(\beta_{x,f,i}, \delta_c)$ and $(\beta_{y,f,i}, \delta_c)$  which connect $\beta_{x,f,i}$ and $\beta_{y,f,i}$ respectively, to the clause node of the $i$-th clause where $x$ ($y$) occurs negated, $i=1,2$.

Summarizing the construction, we have specified the graph $G=(L \cup F,A)$ by choosing:
\begin{eqnarray*}
L & =& \{\{t_{x,i}, f_{x,i}\}|~x \in X, i=1,2\} \cup \{d\},\\
F &=& \{\{\alpha_{x,i}, \alpha_{y,i}\}|~x \in X, y \in Y, i=1,2\} \cup B \cup \{\{\tau_{y,i}, \phi_{y,i}\}|~y \in Y, i=1,2\} \cup \{\delta_c|~c \in C\},\\
A &=& \cup_{x \in X} A_x \bigcup \cup_{y \in Y} A_y \bigcup \{(\delta_c,d), (d, \delta_c)|~c \in C\} \\
&&\cup~ \{(\beta_{x,t,i}, \delta_c), (\delta_c, \beta_{x,t,i})|~c \in C \mbox{ is the }i\mbox{-th clause containing } x,\ x \in X, i=1,2\}
\\
&&\cup~ \{(\beta_{x,f,i}, \delta_c), (\delta_c,\beta_{x,f,i}) |~c \in C \mbox{ is the }i\mbox{-th clause containing } \neg x,\ x \in X, i=1,2\}\\
&&\cup~ \{(\beta_{y,t,i}, \delta_c), (\delta_c,\beta_{y,t,i}) |~c \in C \mbox{ is the }i\mbox{-th clause containing } y,\ y \in Y, i=1,2\}
\\
&&\cup~ \{(\beta_{y,f,i}, \delta_c), (\delta_c, \beta_{y,f,i})|~c \in C \mbox{ is the }i\mbox{-th clause containing } \neg y,\ y \in Y, i=1,2\}.
\end{eqnarray*}

We set $K=3$, meaning that the length of a cycle present in a solution cannot exceed 3. Finally, we set $k = 4|X|+1$. This completes the description of an instance of Stackelberg KEP.\\

$\Rightarrow$ Given a truth assignment for $X$ such that no truth assignment exists for $Y$ that satisfies $E$, we now show the existence of a strategy $S \subseteq L$ such that $w(G[S]) + w^{L}(G[V \setminus S]) \geq 4|X|+1 = k$.

Given a truth assignment for $X$, we propose the following strategy $S$:
\begin{equation}
\label{eq:defS}
    S = \{\{t_{x,1}, t_{x,2}\}|~ x \in X \mbox{ is true}\} \cup \{\{f_{x,1}, f_{x,2}\}|~ x \in X \mbox{ is false}\}.
\end{equation}

In words: for each variable $x \in X$ that is TRUE, $t_{x,i} \in S$ for $i = 1,2$ and for each variable $x \in X$ that is FALSE, $f_{x,i} \in S$ for $i = 1,2$. 

Recall that $|L|=4|X|+1$. Hence, we need to show that given this strategy $S$, each node of the leader is contained in a maximum size $3$-cycle packing of the leader (i.e., a maximum size $3$-cycle packing on $G[S]$), or in each maximum size $3$-cycle packing of the follower (i.e., a maximum $3$-cycle packing on $G[V \setminus S]$). By the choice of $S$, there are $2|X|$ nodes in $S$, all of which are contained in a maximum size $3$-cycle packing of the leader. Indeed, such a maximum size $3$-cycle packing of the leader consists of $|X|$ cycles of length 2, each containing a pair of $t$-nodes or a pair of $f$-nodes of the corresponding variable gadget. It remains to show that every maximum size $3$-cycle packing for $G[V \setminus S]$ contains all $t$ and $f$-nodes not in $S$, as well as the $d$-node.

To do so, we now analyze the possible 3-cycle packings in $G[V\setminus S]$. Notice that any cycle in $G[V\setminus S]$ that contains nodes of different variable gadgets has length more than 3, and hence cannot be present in a $3$-cycle packing. It follows that each cycle in the follower's 3-cycle packing
\begin{itemize}
    \item consists of nodes all within a single variable gadget (called a cycle of Type 1), or
    \item consists of the nodes  $\{\beta, \delta_c\}$ for some $\beta \in B, c \in C$ (called a cycle of Type 2), or
    \item consists of the nodes $\{\delta_c, d\}$ for some $c \in C$ (called a cycle of Type 3).
\end{itemize}

We now classify the variable gadgets with respect to the possible cycles of Type 1 contained in the variable gadget. The classification of these gadgets is illustrated in Figure \ref{fig:followerpackings}.

\begin{definition}
Given a solution to the follower's cycle packing problem, we call a gadget corresponding to variable $x \in X$
\begin{itemize}
\item {\em consistent} if either the two node-sets $\{\alpha_{x,i}, \beta_{x,f,i}, f_{x,i}\}$, $i=1,2$, or the two node-sets $\{\alpha_{x,i}, \beta_{x,t,i}, t_{x,i}\}$, $i=1,2$ each correspond to a cycle in the solution (Fig. \ref{Fig:Xconsistent}),
\item {\em cheating} if the node-set $\{f_{x,1}, f_{x,2}\}$ or $\{t_{x,1}, t_{x,2}\}$ corresponds to a cycle in the solution (Fig. \ref{Fig:XCheat}).
\end{itemize}

Given a solution to the follower's cycle packing problem, we call a gadget corresponding to variable $y \in Y$
\begin{itemize}
\item {\em consistent} if either the node-sets $\{\alpha_{y,i}, \beta_{y,f,i}, \phi_{y,i}\}$, $i=1,2$ as well as the node-set $\{\tau_{y,1}, \tau_{y,2}\}$ each correspond to a cycle in the solution, or if the node-sets $\{\alpha_{y,i}, \beta_{y,t,i}, \tau_{y,i}\}$, $i=1,2$ as well as the node-set $\{\phi_{y,1}, \phi_{y,2}\}$ each correspond to a cycle in the solution (Fig. \ref{Fig:Yconsistent}),
\item {\em cheating} if the two node-sets $\{\tau_{y,1}, \tau_{y,2}\}$ and $\{\phi_{y,1}, \phi_{y,2}\}$ each correspond to a cycle in the solution (Fig. \ref{Fig:YCheat}),
\item {\em zigzag} if the node-sets $\{\alpha_{y,1}, \beta_{y,t,1}, \tau_{y,1}\}$ and $\{\alpha_{y,2}, \beta_{y,f,2}, \phi_{y,2}\}$ each correspond to a cycle in the solution (Fig. \ref{Fig:Zigzag}).
\end{itemize}
\end{definition}

Consistent gadgets will be used to reflect the truth value of the corresponding variables. We say the gadget is consistent with a TRUE value if the two node-sets $\{\alpha_{x,i}, \beta_{x,f,i}, f_{x,i}\}$ for $i = 1,2$ each correspond to a cycle in the follower's cycle packing solution. Analogously, if in the follower's solution the node-sets $\{\alpha_{x,i}, \beta_{x,t,i}, t_{x,i}\}$ for $i = 1,2$, each correspond to a cycle, we say the gadget is consistent with FALSE.\\


\begin{figure*}[t!]
    \centering
    \begin{subfigure}[t]{0.45\textwidth}
        \centering
        \includegraphics[width = 130pt]{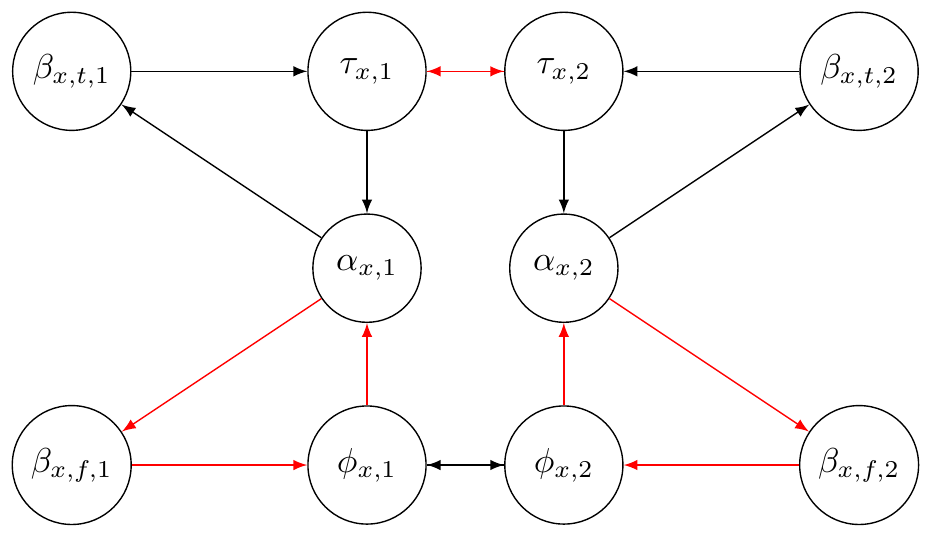}
        \caption{Consistent gadget $(y \in Y)$}\label{Fig:Yconsistent}
    \end{subfigure}%
    ~
    \begin{subfigure}[t]{0.45\textwidth}
        \centering
        \includegraphics[width = 130pt]{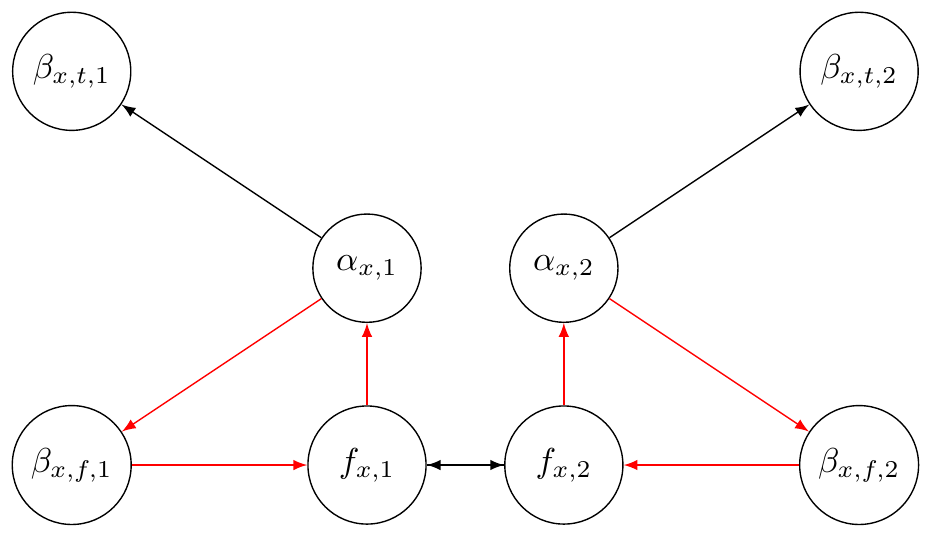}
        \caption{Consistent gadget $(x \in X)$} \label{Fig:Xconsistent}
    \end{subfigure}
    \vskip \baselineskip
    \begin{subfigure}[t]{0.45\textwidth}
        \centering
        \includegraphics[width = 130pt]{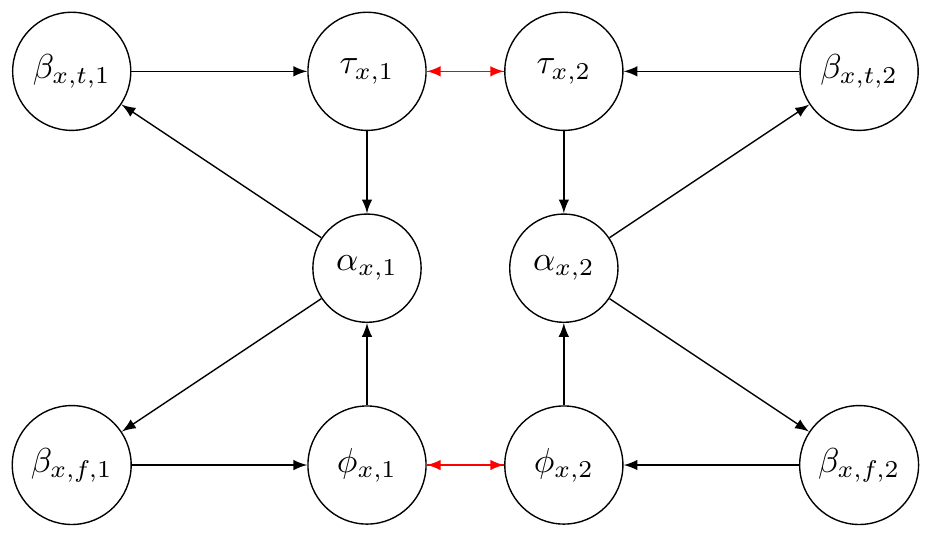}
        \caption{Cheating gadget $(y \in Y)$} \label{Fig:YCheat}
    \end{subfigure}
    ~
    \begin{subfigure}[t]{0.45\textwidth}
        \centering
        \includegraphics[width = 130pt]{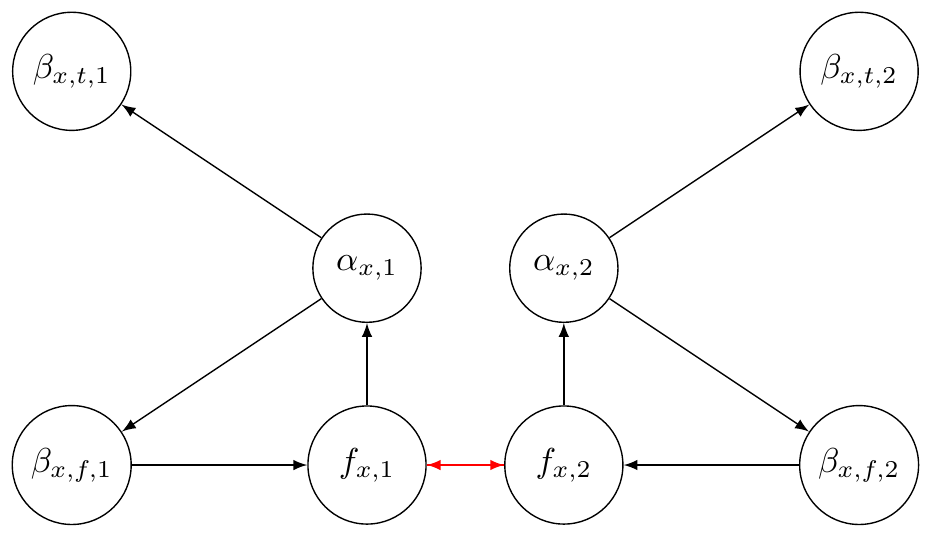}
        \caption{Cheating gadget $(x \in X)$} \label{Fig:XCheat}
    \end{subfigure}
    \vskip \baselineskip
    \begin{subfigure}[t]{0.45\textwidth}
        \centering
        \includegraphics[width = 130pt]{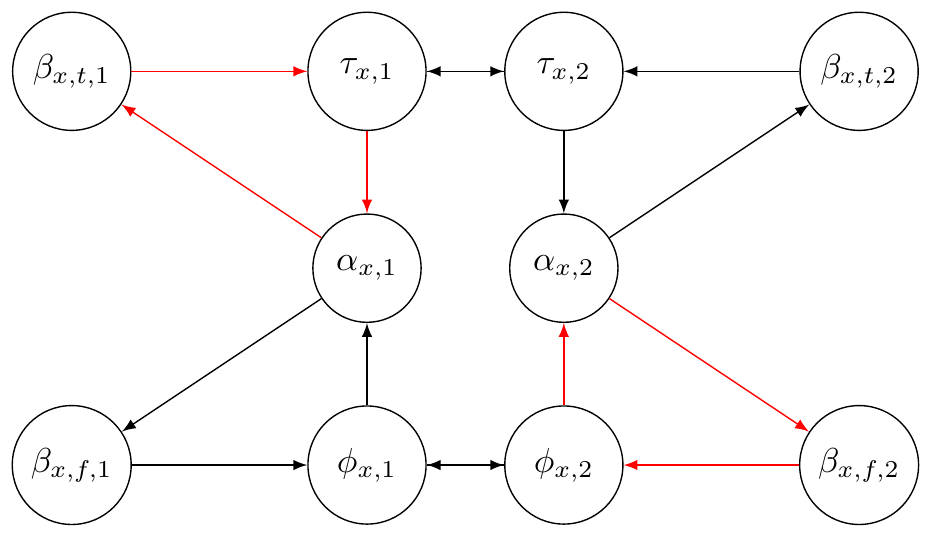}
        \caption{Zigzag gadget $(y \in Y)$} \label{Fig:Zigzag}
    \end{subfigure}
    \caption{Classification of cycle packings on variable gadgets}
    \label{fig:followerpackings}
\end{figure*}

We use these definitions to characterize an optimal 3-cycle packing of the follower, as witnessed by the following lemma.

\noindent
\begin{lemma}
\label{lem:classi}
Let $S$ be defined by (\ref{eq:defS}). In each optimal 3-cycle packing on $G[V \setminus S]$, each variable gadget is either consistent, cheating or zigzag.
\end{lemma}
\begin{proof} We argue as follows. Consider a feasible solution to the follower's cycle packing problem such that there is a variable gadget which is neither consistent, nor cheating, nor zigzag. We show that that solution is not of maximum size by exhibiting the existence of a strictly better solution.

We first consider variable gadgets corresponding to variables $x \in X$. Without loss of generality, we assume that $t_{x,1}, t_{x,2} \in S$, as depicted in Figures \ref{Fig:Xconsistent} and \ref{Fig:XCheat}. There are two cases.

\begin{enumerate}
    \item[Case 1:] Consider a solution such that for some variable gadget corresponding to $x \in X$, at most one of $\{\beta_{x,f,1}, \beta_{x,f,2}\}$ is covered by a cycle of Type 2. Since, by assumption, the gadget is not consistent, at most three of its nodes are covered by a cycle of Type 1.  By removing at most one cycle of Type 2 (thereby ``freeing'' a $\beta$ node), and introducing a cycle of Type 1, we have changed the state of this gadget to consistent. Moreover, the size of the packing has increased by at least $-2+3 = 1$.
    \item[Case 2:] Consider a solution such that for some variable gadget corresponding to $x \in X$, both $\beta_{x,f,i}$, $i = 1,2$ are covered by cycles of Type 2. Since, by assumption, the gadget is not cheating, it follows that no other nodes of the gadget are contained in a cycle. By simply adding the cycle consisting of the nodes $\{f_{x,1}, f_{x,2}\}$, the size of the packing increases by 2.
\end{enumerate}

Due to space restrictions, we include the remainder of the proof, for variable gadgets corresponding to variables $y \in Y$, in the appendix. \qed
\end{proof}

Thus, we have proven that in any optimum 3-cycle packing of the follower, each gadget is either consistent, cheating or zigzag. Given this structure of any optimal solution of the follower, we will now argue that the leader's strategy $S$ ensures that all $4|X|+1$ leader nodes will be covered (if the truth assignment on $X$ is such that there exists no truth assignment on $Y$ satisfying $E$). Observe that, given the two possible packings on variable gadgets corresponding to variables $x \in X$, all leader nodes within the variable gadgets are covered. Thus, we have covered already $4|X|$ nodes of the leader. We must now show that the $d$ node will also be covered in any maximum size packing of the follower.

By Lemma~\ref{lem:classi}, we know that an optimal solution of the follower either contains a variable gadget that is cheating or zigzag, or all variable gadgets are consistent. We proceed to argue that in each of these two cases, any optimal solution of the follower covers node $d$.

\begin{lemma}
\label{lem:noded}
Let $S$ be defined by (\ref{eq:defS}). In each optimal 3-cycle packing on $G[V \setminus S]$, node $d$ is covered.
\end{lemma}

\begin{proof}
We distinguish two cases. In both cases, we argue by contradiction, i.e., we argue that if vertex $d$ is not covered in a 3-cycle packing, that 3-cycle packing is not of maximum size. 

\begin{enumerate}
    \item [Case 1:] The follower's solution uses a gadget that is either cheating or zigzag. Let us further suppose that this follower's solution is such that node $d$ is not covered.
\begin{itemize}
    \item Cheating gadget: the cheating gadget covers 4 nodes by cycles of Type 1, compared to 8 nodes in a consistent gadget. Switching to a consistent packing is thus always strictly better unless the cheating packing allows for 2  additional cycles of Type 2. This is only the case if all four $\beta$-nodes are covered by such cycles in the cheating gadget. Indeed, if less than four are covered, the packing can be set consistent to the truth value breaking a minimum cycles of Type 2. However, even if all four $\beta$-nodes are covered by cycles of Type 2, this still only leads to parity between the consistent and cheating gadget (since they both cover 12 nodes). Thus, if the $d$ node is uncovered, one of the cycles of Type 2 can be replaced by the cycle $(\delta_c, d)$ of Type 3, and the consistent packing achieves 14 nodes over the variable and linked clause gadgets.
    \item Zigzag gadget: by an analogous argument, it can be shown that if $d$ is uncovered, switching from a zigzag to a consistent packing increases the number of covered nodes by two.
\end{itemize}
    \item[Case 2:] The follower's solution uses only consistent gadgets. Let us further suppose that this follower's solution is such that node $d$ is not covered. Recall that the truth assignment on $X$ used to build the strategy $S$ is such that there does not exist any truth assignment on $Y$ satisfying $E$.

    Since node $d$ is not covered, it follows immediately that in this cycle packing, each $\delta_c$ is covered by a cycle of Type 2. If this were not the case, the cycle $(\delta_c, d)$ would strictly increase the size of the packing. Now let us build a truth assignment for the (2,2)-SAT instance based on the packing. If the packing restricted to a variable gadget is consistent with TRUE (FALSE), set that variable to TRUE (FALSE). Clearly, for each variable $x \in X$ this truth assignment is the same as the original truth assignment used to construct the strategy $S$. We claim that this truth assignment satisfies each clause. Indeed, for a given clause $c$ let there be (wlog) the cycle $(\delta_c, \beta_{y,t,i})$. By construction of the Stackelberg KEP instance these arcs only exist if the clause is satisfied by a value of TRUE for $y$. Furthermore, by construction of the truth assignment, we have set $y$ to true. Thus, the truth assignment is such that every clause is satisfied. Thus, we have arrived at a contradiction.
\end{enumerate}
Since Lemma~\ref{lem:classi} implies there are no other cases, the result follows. \qed
\end{proof}

Summarizing, if there exists a truth assignment to $X$ such that there does not exist a truth assignment to $Y$ satisfying $E$, the leader can construct a strategy $S$. This strategy is such that if the follower's solution uses only consistent gadgets, there is at least one $\delta_c$-node that can not be covered by a cycle of Type 2, and will thus be covered by a cycle with node $d$, i.e., a cycle of Type 3. Alternatively, if the follower's solution contains a cheating or a zigzag gadget then it must also be the case that node $d$ is covered. The leader is thus guaranteed that the strategy $S$ implies that all its $4|X|+1$ nodes are covered.\\[2ex]



\noindent
$\Leftarrow$ Suppose that there exists a strategy $S$ such that the leader can guarantee that all its $4|X|+1$ nodes are covered. We will show that the existence of such a strategy implies that there exists a truth assignment for $X$ such that there is no truth assignment for $Y$ satisfying the expression $E$.

We first analyze the structure of any strategy $S$ that guarantees covering all $4|X|+1$ nodes of the leader; we use $\mathcal{S}$ to denote the collection of strategies that guarantee that all nodes of the leader are covered.


Obviously, for each $S \in \mathcal{S}$ it must hold that either both $t_{x,1}$ and $t_{x,2}$, or none of $t_{x,1}$ and $t_{x,2}$ are in $S$ for each $x \in X$. Indeed, if $S$ contains exactly one node from $\{t_{x,1}, t_{x,2}\}$ for some gadget corresponding to $x \in X$, it is impossible to cover that node of the leader. Thus, $\mathcal{S}$ contains only strategies $S$ for which either both $t_{x,1}$ and $t_{x,2}$, or none of $t_{x,1}$ and $t_{x,2}$ are in $S$ for each $x \in X$. The same statement holds for the nodes $f_{x,1}$ and $f_{x,2}$, for some $x \in X$: either both nodes $f_{x,1}, f_{x,2}$ are in $S$ or none of them, for each $x \in X$.

Further, we call a strategy $S$ {\em nice} if, for each $x \in X$, either $t_{x,1}, t_{x,2} \in S$, or $f_{x,1}, f_{x,2} \in S$ but not both.

The following lemma describes the presence of this property in optimal strategies.

\begin{lemma}
\label{lem:Sisnice}
There exists an optimal strategy $S \in \mathcal{S}$ that is nice.
\end{lemma}
\begin{proof}
We argue by contradiction. Thus, suppose that each $S \in \mathcal{S}$ is not nice. It follows that there exists a set of gadgets corresponding to $W \equiv W_1 \cup W_2 \subseteq X$ such that
\begin{enumerate}
    \item $W_1 := \{x \in X \mid  t_{x,1}, t_{x,2}, f_{x,1}, f_{x,2} \notin S\}$, and
    \item $W_2 := \{x \in X \mid t_{x,1}, t_{x,2}, f_{x,1}, f_{x,2} \in S\}$.
\end{enumerate}

We now prove that, given some $S \in \mathcal{S}$, we can construct a strategy $S' \in \mathcal{S}$ which is nice, thereby proving the lemma. Indeed, given some $S \in \mathcal{S}$, let strategy $S'$ identical to $S$ except that, for each $x \in W$, we set $t_{x,1}, t_{x,2} \in S'$ and $f_{x,1}, f_{x,2} \notin S'$. Clearly, it follows that $S'$ is nice. We proceed to argue that all $4|X|+1$ nodes of the leader are guaranteed to be covered by $S'$, i.e., $S' \in \mathcal{S}$.

We argue by contradiction. Suppose this is not the case, i.e., the constructed nice strategy $S'$ cannot guarantee that all $4|X|+1$ leader nodes are covered. Then there exists a maximum size cycle packing on $G[V \setminus S']$ for the follower which does not cover $d$. The strategy $S'$ being nice, and cycle packing on $G[V \setminus S']$ being of maximum size while not covering $d$, implies the following:
\begin{enumerate}
    \item Each clause node $\delta_c$ ($c \in C$) is in a cycle of Type 2 with a $\beta$-node, and
    \item each gadget is consistent; this follows from Lemma~\ref{lem:classi}, and the proof of Lemma~\ref{lem:noded}, case 1.
\end{enumerate}

Given the cycle packing on $G[V \setminus S']$, we now construct a cycle packing on $G[V \setminus S]$. First, these packings are identical with respect to the cycles of Type 2, and the cycles of Type 1 in the gadgets corresponding to $y \in Y$ and $x \in X \setminus W$. Note that these gadgets are all consistent. Next, for gadgets corresponding to a variable $x \in W_1$, choose cycles such that the gadget is consistent with True.  For the gadgets of variables $x \in W_2$, the follower can not choose any cycle of Type 1 in the gadget.

The size of the cycle packing on $G[V \setminus S]$ as just described is $2|C| + 8|Y| + 8|W^1| + 6|X \setminus W|$. This is also an upper bound on the size of any maximum cycle packing on $G[V \setminus S]$. Indeed, since each cycle of Type 2 or Type 3 has a length of 2 and covers one node $\delta_c$, $c \in C$, their combined size is at most $2|C|$. The size of cycles of Type 1 per gadget is similarly bounded by $8$ for gadgets corresponding to $y \in Y$ and $x \in W^1$, $6$ for gadgets corresponding to $x \in X \setminus W$, and no cycles can be chosen in gadget corresponding to $x \in W^2$. Since the size of the constructed matching on $G[V \setminus S]$ matches its upper bound, it is a maximum matching. Given a maximum size cycle packing on $G[V \setminus S']$ that does not cover $d$, we can thus construct a maximum size cycle packing on $G[V \setminus S]$, a contradiction.\qed
\end{proof}

By Lemma~\ref{lem:Sisnice}, we are ensured that there is a nice strategy $S$ within the class $\mathcal{S}$. Given a nice strategy $S$, we formulate the solution to Adversarial (2,2)-SAT accordingly: if $t_{x,1}, t_{x,2} \in S$ and $f_{x,1}, f_{x,2} \notin S$, set $x \in X$ to TRUE, and conversely, if $t_{x,1}, t_{x,2} \notin S$ and $f_{x,1}, f_{x,2} \in S$, set $x \in X$ to FALSE.
We claim that  this truth assignment for $X$ is such that there does not exist a truth assignment for $Y$ satisfying $E$.
Indeed, if a truth assignment existed for $Y$ satisfying $E$ there would exist a solution to the follower's problem consisting of consistent gadgets only, reflecting the truth assignment for $Y$, such that each $\delta_c$ is covered by a cycle with a $\beta$-node ($c \in C, \beta \in B$), leaving $d$ uncovered.
This finishes the proof.\qed
\end{proof}

\section{The Stackelberg KEP game with $K = 2$}
In this section, we consider Stackelberg KEP for $K=2$. We will show that the Stackelberg KEP game for $K = 2$ is polynomially solvable, i.e., we can compute an optimal strategy $S$ for the leader.

The proof of this claim heavily relies on the results by \cite{carvalho2017} and \cite{carvalho2019game}. We will show that the leader's optimal strategy can be determined by solving the problem of computing a player's best reaction in an $N$-KEG game, whenever the strategies of the other $N-1$ players are considered fixed. Note that in an $N$-KEG game, the players are restricted to play a strategy in which they contribute all internally unmatched pairs to the common pool. This stands in contrast to the setting of the Stackelberg KEP game, where the leader is allowed to withhold unmatched nodes from the follower.

In the following lemma, we show that contributing an extra node to the common pool never decreases the minimum number of leader nodes matched in a maximum size matching. As a result, strategies where the leader does not contribute one or more nodes not covered by the internal packing is (weakly) dominated by the strategy with an identical internal packing where the leader contributes all nodes that are not covered. Thus, there always exists an optimal strategy where the leader contributes all nodes that are not covered to the common pool.

For ease of notation, 
we reduce the directed compatibility graph to an undirected graph $G = (V,E)$, where $E$ consists of edges $\{u,v\}$ for which the arc set of the directed counterpart contains both $(u,v)$ and $(v,u)$.

\begin{lemma}\label{lemma:weak_dominance}
Let $G = (V = L \cup F, E)$ be an undirected graph, $S \subseteq L$ a strategy of the leader and $u \in S$ a node. Then:
\[w^L(G[V \setminus S]) \le w^L(G[(V \setminus S) \cup \{u\}]).\]
\end{lemma}
Due to space restrictions, we include the proof of this lemma in the appendix. \\

Lemma~\ref{lemma:weak_dominance} shows that whenever strategy $S \subseteq L$ is chosen and $u \in S$ is unmatched with respect to a maximum matching on $G[S]$, there is no incentive for the leader to hide node $u$ from the follower. Therefore, the leader can restrict itself to strategies $S \subseteq L$ for which $G[S]$ allows a perfect matching.

Furthermore, we notice that in contrast to the setting of the $N$-KEG problem in~\cite{carvalho2019game} where the independent agent is not allowed to use edges between nodes of the same player (internal edges) the Stackelberg KEP game does not have this restriction. Once again, we claim that for any strategy $S' \subseteq L$ for which the follower will choose an internal edge $\{u,v\} \subseteq L$ in the second phase of the Stackelberg KEP game, there exists a weakly dominating strategy $S \subseteq L$ for which the follower will not pick internal leader edges on the maximum size matching on $G[V \setminus S]$.

\begin{lemma}\label{lemma:no_internal_edges}
Let $G = (V = L  \cup  F, E)$ be an undirected graph. There exists an optimal strategy $S \subseteq L$ such that the follower chooses a maximum size matching on $G[V \setminus S]$ with no internal leader edges.
\end{lemma}
Due to space restrictions, we include the proof of this lemma in the appendix.\\

Notice that whenever we impose the hard constraint that the follower is not allowed to use internal leader edges, the minimum number of covered leader nodes in a maximum size follower matching can never increase. Therefore, it also follows that the strategic options of the follower are really equivalent to those of an independent agent in a suitably constructed $N$-KEG game. Together with this observation, we now have all the necessary tools to derive the complexity of the Stackelberg KEP game restricted to pairwise kidney exchanges only.
\begin{theorem}
The Stackelberg KEP game is polynomially solvable if the maximum cycle length $K = 2$.
\end{theorem}
Due to space restrictions, we include the proof of this theorem in the appendix.

\section{Conclusion}
Collaboration between hospitals and countries (referred to as agents) has the potential to improve the number of kidney transplants. However, goals of individual agents may not lead social optima. Research on this topic has focused on strategy-proof mechanisms and the loss of transplants associated with individual rational strategies. In specific cases, with limited cycle length, it can be shown that solutions which are both socially optimal and individually rational exist and can be identified efficiently.

In this paper, we show that for an individual agent, the problem of maximizing its own transplants can computationally be very complex, even given perfect knowledge of other agent's actions and the allocation mechanism used in the common pool. This result has the potential to  weaken the need for deploying mechanisms that are strategy-proof, as individual agents will find it challenging to identify their own optimal strategy. 

\bibliographystyle{plainnat}
\bibliography{KEP}

\section*{Appendix - Additional Proofs}
\subsection*{Remainder of proof of Lemma 1}
Now we consider variable gadgets corresponding to variables $y \in Y$. For ease of exposition, we assume that for such a gadget in the current solution, at least as many $\beta_{y,t,i}$ as $\beta_{y,f,i}$ ($i=1,2$) are covered by cycles of Type 2. This is without loss of generality. We consider the four $\beta$ nodes of the gadget, and make a case distinction based on whether these $\beta$ nodes are covered by cycles of Type 2.
\begin{enumerate}
\item[Case 1:] All four $\beta$ nodes of the gadget are covered by cycles of Type 2. Since, by assumption, the gadget is not cheating, one easily verifies that the size of the packing improves when the solution is changed such that the gadget becomes a cheating gadget.
\item[Case 2:] Three of the four $\beta$ nodes of the gadget are covered by cycles of Type 2. Hence, a single $\beta$ node is not covered by a cycle of Type 2, say $\beta_{y,f,2}$. It also follows that $\alpha_{y,1}$ is not covered. If $\phi_{y,1}$ is not covered, we remove the cycle of Type 2 covering $\beta_{y,f,1}$, and add the cycle of Type 1 containing nodes $\beta_{y,f,1}, \phi_{y,1}$ and $\alpha_{y,1}$, thereby increasing the size of the packing. If $\phi_{y,1}$ is covered, it is in a cycle with $\phi_{y,2}$, meaning that $\alpha_{y,2}$ is not covered. Then, we remove the cycle of Type 1 covering $\phi_{y,1}$ and $\phi_{y,2}$, as well as the cycle of Type 2 covering $\beta_{y,f,1}$, and we add the two cycles of Type 1 containing nodes $\beta_{y,f,i}, \phi_{y,i}$ and $\alpha_{y,i}$, $i=1,2$, again increasing the size of the packing.
\item[Case 3:] Two of the four $\beta$ nodes are covered by cycles of Type 2. If these two nodes are $\beta_{y,t,i}$, $i=1,2$, then it is easy to see that by modifying the solution such that the gadget becomes consistent (in fact, consistent with a TRUE value) is the only possibility since all nodes of the gadget are then covered. If $\beta_{y,t,1}$ and $\beta_{y,f,1}$ are covered by cycles of Type 2, it follows that at most 5 nodes of the gadget can be covered by cycles of Type 1. By removing one cycle of Type 2, this solution can be improved by such that the gadget becomes consistent. The number of nodes covered by cycles of Type 1 rises to 8 and one of the cycles of Type 2 covering $\beta_{y,t,1}, \beta_{y,f,1}$ can still be used. If $\beta_{y,t,1}$ and $\beta_{y,f,2}$ are covered by cycles of Type 2, and since, by assumption, the gadget is not zigzag, at most 5 nodes of the gadget are covered by cycles of Type 1. By switching to a zigzag packing, 6 nodes are covered by cycles of Type 1, and no existing cycles of Type 2 are impacted. The size of the packing increases.
\item[Case 4:] At most one $\beta$ node is covered by a cycle of Type 2. Since, by assumption, the gadget is not consistent, we can improve the solution by modifying this gadget to be consistent.
\end{enumerate}

\subsection*{Proof of Lemma \ref{lemma:weak_dominance}}
\begin{proof}
Clearly, $w(G[V \setminus S]) \le w(G[(V \setminus S) \cup \{u\}]) \le w(G[V \setminus S]) + 1$, we can restrict ourselves to a case distinction with two cases:
\begin{itemize}
    \item [Case 1:] $ w(G[(V \setminus S) \cup \{u\})]) = w(G[V \setminus S]) + 1$. In this case, any maximum matching of $G[(V \setminus S) \cup \{u\}]$ matches $u$. Take an arbitrary maximum matching $M_u$ of $G[(V \setminus S) \cup \{u\}]$, let $e = \{u, v\} \in M_u$ for some $v \in V \setminus S$ be the edge that matches $u$. The matching $M_u \setminus e$ is a maximum matching of $G[V \setminus S]$ with fewer leader nodes (as $u$ is a leader node). Thus, any maximum matching of $G[(V \setminus S) \cup \{u\}]$ corresponds to a maximum matching of $G[V \setminus S]$ covering fewer leader nodes. In particular, this implies $w^L(G[V \setminus S]) < w^L(G[(V \setminus S) \cup \{u\}])$, meaning that we actually increase our objective value by revealing $u$.
    \item [Case 2:] $w(G[(V \setminus S) \cup \{u\}]) = w(G[V \setminus S])$. Clearly, we have that $w^L(G[V \setminus S]) \ge w^L(G[(V \setminus S) \cup \{u\}])$: any maximum matching of $G[V \setminus S]$ is also a feasible maximum matching for $G[(V \setminus S) \cup \{u\}]$. We will now show by contradiction that also $w^L(G[V \setminus S]) \le w^L(G[(V \setminus S) \cup \{u\}]) $ must hold.\\
    \\
    Suppose that $w^L(G[(V \setminus S) \cup \{u\}]) < w^L(G[V \setminus S])$. In that case, let us consider a maximum matching $M_u$ of $G[(V \setminus S) \cup \{u\}]$ covering $w^L(G[(V \setminus S) \cup \{u\}])$ leader nodes. Then $u$ is matched in $M_u$, as otherwise $M_u$ would be a maximum matching on $G[V \setminus S]$ with fewer than $w^L(G[V \setminus S])$ covered leader nodes. Let $e = \{u, v\} \in M$ for some $v \in V \setminus S$. The matching $M' = M_u \setminus e$ is a $(w(G[V \setminus S])-1)$-cardinality matching in $G[V \setminus S]$, thus non-maximum. This implies that there exists an $M'$-augmenting path in $G[V \setminus S]$ due to~\cite{Berge1957}.\\
    \\
    We claim that any $M'$-augmenting path starts in $v$; if not, let $P$ be an $M'$-augmenting path in $G[V \setminus S]$ not starting in $v$. Then: $M' \oplus P$ is a $w(G[V \setminus S])$-cardinality matching in $G$, where $A \oplus B := \left(A\setminus B\right) \cup \left(B \setminus A\right)$ denotes the symmetric difference of $A$ and $B$. Then, the nodes $u$ and $v$ are still both unmatched, as we assumed $P$ does not contain $v$. This implies that $\left(M' \oplus P\right) \cup \left\{\{u,v\}\right\}$ is a $(w(G[V \setminus S]) + 1)$-cardinality matching in $G[V \setminus (S \cup \{u\})]$, a contradiction.\\
    \\
    Therefore, any $M'$-augmenting path $P = \{v = v_0, v_1, \ldots, v_k = w\}$ in $G[V \setminus S]$ must start in $v$. The set of matched nodes in the $w(G[V \setminus S])$-cardinality matching $M' \oplus P$ is almost the same as the set of matched nodes in $M_u$, except $u$ in $M_u$ is exchanged for node $w$ (the endnode of $P$ unequal to $v$). This means $M' \oplus P$ is a maximum matching on $G[V \setminus S]$ and covers either $w^L(G[(V \setminus S) \cup \{u\}])$ (when $w \in L$) or $w^L(G[(V \setminus S) \cup \{u\}])-1$ (when $w \in F$) leader nodes, a contradiction.
\end{itemize}
Thus, in both cases, we obtain $w^L(G[V \setminus S]) \le w^L(G[(V \setminus S) \cup \{u\}])$, which finishes the proof.\qed
\end{proof}

\subsection*{Proof of Lemma 5}
\begin{proof}
Let $S' \subseteq L$ be an arbitrary feasible leader strategy. Let $M$ be a maximum size matching on $G[V \setminus S']$ covering exactly $w^L(G[V \setminus S'])$ leader nodes. Let $N = M \cap E(G[L])$ be the submatching of $M$ consisting of the internal leader edges. Consider now the feasible strategy $S = S'  \cup  V(N)$. Strategy $S$ has the same guaranteed objective value as $S'$, but now the follower does not pick any internal leader edges anymore. This shows that there exists an optimal strategy for the leader in which the follower will not pick any internal leader edges. \qed
\end{proof}

\subsection*{Proof of Theorem 2}
\begin{proof}
Consider a Stackelberg KEP game on the graph $G =(V = L \cup F, A)$ and maximum cycle length $K=2$. We now construct a $(|F| + 1)$-KEG with maximum cycle length $K=2$. Player one controls the node set $L$, while the remaining $|F|$ players each control a unique node from the set $F$. The role of the follower is to compute a maximum matching on the graph induced by all nodes contributed to the common pool, such that a minimum number of player one's nodes are included. \\

In the $(|F| + 1)$-KEG, each of the players with only node has only one feasible strategy, not to match any nodes internally. These players have no incentive to hide its unique node from the common pool, as then the player is guaranteed to have no matched donor-recipient pairs. Meanwhile, player one is faced with the problem of finding a matching $M$ maximizing the number of nodes in $L$ which are matched either internally or by the follower. Note that the number of nodes matched internally is $w(G[M])$. Due to the choice of the matching algorithm of the independent agent, and the fact that every player except for player one must contribute all nodes to the common pool, the number of player one's nodes matched externally is also $w^L(G[V \setminus M])$. The objective of player one is thus the same as the objective of the leader in the Stackelberg KEP game. Lemma \ref{lemma:weak_dominance} shows that given this objective function, there exists an optimal strategy $S$ for the leader in the Stackelberg KEP game for which there exists a perfect matching on $G[S]$. From Lemma~\ref{lemma:no_internal_edges} it follows that the follower in the Stackelberg KEP game operates on the same graph as the follower in the $N$-KEG game, namely the graph induced by the common pool and with only noninternal edges. Thus, the optimal strategy $M$ for player one is also an optimal strategy for the leader in the Stackelberg KEP game. \cite{carvalho2019game} prove that computing a Nash equilibrium in the $N$-KEG can be done in polynomial time for any deterministic algorithm for the independent agent. Notice that computing the optimal strategy of player one reduces to finding an optimal reaction to the strategies of the $|F|$ other players, each of which only have one feasible strategy, namely to reveal their unique node. Thus, the strategy for the leader in the Stackelberg game with pairwise exchanges can be computed in polynomial time, as the best reaction of player one together with the fixed strategies of the other players forms a Nash equilibrium, which clearly implies that the decision problem can be answered in polynomial time as well. \qed \end{proof}
\newpage
\subsubsection*{Proof of $\Sigma_2^p$-completeness of Adversarial (2,2)-SAT}

We show that Adversarial (2,2)-SAT is $\Sigma_2^p$-complete using the following theorem by \cite{johannes}.

\begin{theorem} [\cite{johannes}]
Let (ADV) be an adversarial problem. Let (NON-ADV) denote the corresponding non-adversarial problem. We assume that (NON-ADV) is in NP. Let $f$ be a polynomial transformation from 3-SAT to (NON-ADV) that satisfies the following property. If $U$ is the ground set of an instance $I_{(3SAT)}$ of 3-SAT and $Z$ is the ground set of the instance $f(I_{(3SAT)})$ of (NON-ADV), then there is a subset $Z'$ of $Z$ and a bijective function $g : U \rightarrow Z'$ such that:
\begin{enumerate}
    \item If $S^U$ is a satisfying solution of $I_{(3SAT)}$, then the 0-1 assignment $S_{Z'}$ to the variables in $Z'$ with $S_{Z'}(z)$ = $S^{U}(g^{-1}(z))$ for all $z \in Z'$ can be extended to a 0-1 assignment $S$ of all variables in $Z$ such that $S_Z$ is a satisfying solution.
    \item If $S_Z$ is a satisfying solution of $f(I_{(3SAT)})$, then the 0-1 assignment $S^U$ with $S^U(x) = S_Z(g(x))$ for all $x \in U$ represents a satisfying solution to $I_{(3SAT)}$.
\end{enumerate}
Then (ADV) is $\Sigma_2^p$-complete.
\end{theorem}
Given a instance of 3-SAT, let us now construct a (2,2)-SAT instance. Let $u_i \in U$ be a variable in the 3-SAT instance with $k$ occurrences in the clauses. We then construct variables $z_i^1, \ldots, z_i^k \in Z$. We furthermore construct the clauses $\{(z_i^1, \neg z_i^2), \ldots, (z_i^{k-1}, \neg z_i^k),(z_i^k, \neg z_i^1)\}$. We refer to these as the {\it consistency clauses}. Note immediately that for every satisfying truth assignment in the (2,2)-SAT instance, all variables $z_i^1,\ldots,z_i^k$ must have the same truth value.
Furthermore, for each clause $c \in C$ in the 3-SAT instance, we construct a clause $c'$ in the (2,2)-SAT instance. Let the occurrence of $u_i$ in $c$ be the $k$-th occurrence (negated and unnegated combined) of $u_i$, then $z_i^k \in c'$ if $u_i$ appears unnegated or $\neg z_i^k \in c'$ if $u_i$ appears negated.
Note that each variable $z_i^j \in Z$ now occurs exactly three times; once negated and once unnegated in the clauses $(z_i^j, \neg z_i^{j+1})$ and once negated OR unnegated in the other clauses. Given that we are reducing to (2,2)-SAT, each variable must occur exactly twice unnegated and twice negated. We construct one more clause, the {\it rest clause}, which contains all these other occurrences. As long as there is at least one variable in the 3-SAT instance that occurs at least once unnegated and at least once negated, this additional clause is automatically satisfied if the consistency clauses are satisfied.

Now, let $Z' = \{z_1^1, z_2^1, \ldots z_n^1\}$ and let $g$ be such that $z_i^1$ has the same truth value as $u_i$. The remainder follows easily. We extend the truth assignment to $Z'$ to $Z$ by setting the truth assignment of $z_i^j$ equal to that of $z_i^1$ for all $i$ and $j$. The consistency clauses and the rest clause are then satisfied (given a non-trivial instance of 3-SAT). Given the identical truth values for all $z_i^j$ given a fixed $i$, and because these truth values are equal to that of $u_i$, the construction of the remaining clauses also ensures they are satisfied. We thus have a satisfying solution for the (2,2)-SAT instance. The other direction can be argued nearly identically.
\end{document}